\documentclass[sigconf]{acmart}
\AtBeginDocument{%
  }

\usepackage[utf8]{inputenc} 
\usepackage[T1]{fontenc}    
\usepackage[table]{xcolor}
\usepackage{hyperref}       
\usepackage{url}            
\usepackage{booktabs}       
\usepackage{amsfonts}       
\usepackage{nicefrac}       
\usepackage{microtype}      
\usepackage{xcolor}         
\usepackage{multirow}
\usepackage{microtype}
\usepackage{graphicx}
\usepackage{booktabs} 
\usepackage{makecell}
\usepackage{array}
\usepackage{float}
\usepackage{algorithm}
\usepackage{algpseudocode}
\usepackage{amsmath}
\usepackage{mathtools}
\usepackage{amsthm}
\usepackage{minted}
\usepackage{listings}
\usepackage{enumitem}
\usepackage{subfig}
\usepackage{shadowtext}
\usepackage{balance}

\copyrightyear{2026}
\acmYear{2026}
\setcopyright{cc}
\setcctype{by}
\acmConference[KDD '26]{Proceedings of the 32nd ACM SIGKDD Conference on Knowledge Discovery and Data Mining V.1}{August 09--13, 2026}{Jeju Island, Republic of Korea}
\acmBooktitle{Proceedings of the 32nd ACM SIGKDD Conference on Knowledge Discovery and Data Mining V.1 (KDD '26), August 09--13, 2026, Jeju Island, Republic of Korea}
\acmPrice{}
\acmDOI{10.1145/3770854.3780296}
\acmISBN{979-8-4007-2258-5/2026/08}

\begin{document}

\title[AbFlow]{AbFlow : End-to-end Paratope-Centric Antibody Design by Interaction Enhanced Flow Matching}

\author{Wenda Wang}
\email{wangwenda87@ruc.edu.cn}
\orcid{0000-0003-2469-8522}
\affiliation{%
  \institution{Renmin University of China, Beijing, China}
  \country{}
}

\author{Yang Zhang}
\email{fengyuewuya@ruc.edu.cn}
\orcid{0009-0005-0208-2964}
\affiliation{%
  \institution{China National Institute of Standardization, Beijing, China}
  \country{}
}
\affiliation{%
  \institution{Renmin University of China, Beijing, China}
  \country{}
}

\author{Zhewei Wei}
\email{zhewei@ruc.edu.cn}
\orcid{0000-0003-3620-5086}
\affiliation{%
  \institution{Renmin University of China, Beijing, China}
  \country{}
  \thanks{Zhewei Wei is the corresponding author. The work was partially done at Gaoling School of Artificial Intelligence, Beijing Key Laboratory of Research on Large Models and Intelligent Governance, Engineering Research Center of Next-Generation Intelligent Search and Recommendation, MOE, and Pazhou Laboratory (Huangpu), Guangzhou, Guangdong 510555, China.}
}

\author{Wenbing Huang}
\email{hwenbing@126.com}
\orcid{0000-0002-2566-4159}
\affiliation{%
  \institution{Renmin University of China, Beijing, China}
  \country{}
}

\renewcommand{\shortauthors}{Wenda Wang, Yang Zhang, Zhewei Wei and Wenbing Huang}

\begin{abstract}
Antigen-antibody binding is a critical process in the immune response. Although recent progress has advanced antibody design, current methods lack a generative framework for end-to-end modeling of full-atom antibody structures and struggle to fully exploit antigen-specific geometric information for optimizing local binding interfaces and global structures. To overcome these limitations, we introduce AbFlow, a paratope-restricted one-step flow-matching framework for designing full-atom antibodies end-to-end. AbFlow incorporates an extended velocity field network featuring an equivariant Surface Multi-channel Encoder, which uses surface-level antigen interaction data to refine the antibody structure, particularly the CDR-H3 region. Extensive experiments in paratope-centric antibody design, multi-CDRs and full-atom antibody design, binding affinity optimization, and complex structure prediction show that AbFlow produces superior antigen-antibody complexes, especially at the contact interface, and markedly improves the binding affinity of generated antibodies.
\end{abstract}





\maketitle
\newcommand\kddavailabilityurl{https://doi.org/10.5281/zenodo.18085636}
\ifdefempty{\kddavailabilityurl}{}{
\begingroup\small\noindent\raggedright\textbf{Resource Availability:}\\
Our code is also available at \url{https://github.com/WangWenda87/AbFlow.git} for direct access and updates.
\endgroup
}

\section{Introduction}
Antibodies can specifically recognize and bind to antigens, thus playing multiple immune defense functions, such as neutralizing the toxicity of pathogens and marking pathogens for clearance by immune cells \citep{buchwalow2009antibodies}. The design of antibodies for specific antigens is significant in biology and immunology \citep{almagro2018progress, douglass2021bispecific, swartz2012engineering}. However, due to the key binding role and high variability of the complementarity-determining regions (CDRs) of antibodies, this task is often challenging \citep{wang2024dna, fu2022antibody, wang2023mesoscale}.

Antibody design involves the joint generation of an antibody's complete sequence and full-atom structure. Traditional energy-based optimization \citep{choi2015antibody, li2014optmaven} was initially applied. Later, 1D sequence learning-based language models \citep{saka2021antibody, liu2020antibody} emerged, characterized by their ability to model antibody sequences linearly, focusing on capturing patterns and relationships within amino acid sequences without considering the 3D structural context.

\begin{figure*}[ht]
\begin{center}
\centerline{\includegraphics[width=1.95\columnwidth]{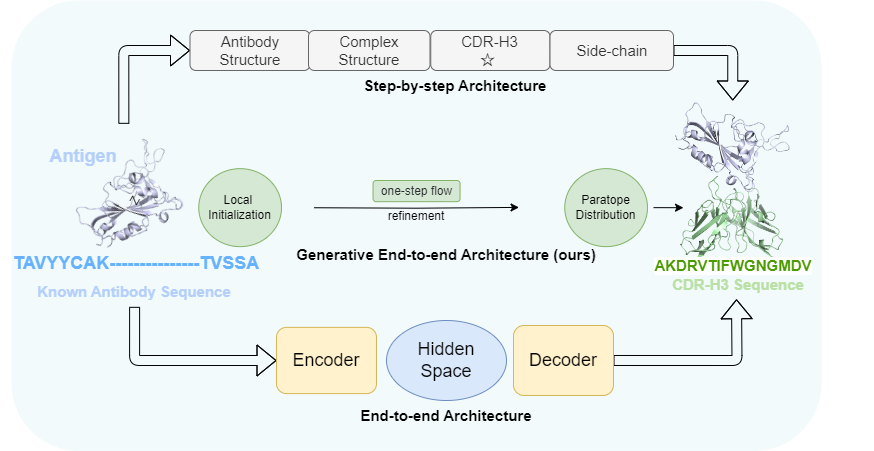}}
\caption{\textbf{Overall illustration of antibody design algorithm.} Given antibody sequences lacking CDR-H3, generate complete sequences and 3D structures. Unlike existing methods that mostly focus on a single step (e.g., CDR-H3 generation) in antibody design tasks and a few end-to-end frameworks, we propose a novel generative end-to-end architecture.}
\Description{Overall illustration of antibody design algorithm.}
\label{fig1}
\end{center}
\end{figure*}

Recent antibody design approaches can be broadly classified into step-by-step pipelines and end-to-end generative methods.
Step-by-step methods decompose the antibody design task into sequential stages, with models tailored to specific sub-problems—typically focusing on CDR design, particularly CDR-H3. These methods generally follow a fixed pipeline: starting from predicting the antibody structure from its sequence \citep{ruffolo2023fast}, docking the antibody and antigen structures \citep{pierce2011accelerating, yan2020hdock}, then generating or predicting CDR-H3 sequences and structures \citep{kong2022conditional, jin2022antibody, luo2022antigen, zhu2024antibody}, and finally performing side-chain packing based on residue identity and geometry \citep{alford2017rosetta}. Among these, MEAN \citep{kong2022conditional} and HERN \citep{jin2022antibody} adopt GNN-based predictive frameworks, while DiffAb \citep{luo2022antigen} and AbX \citep{zhu2024antibody} utilize diffusion-based generative modeling for CDR synthesis.

In contrast, end-to-end methods \citep{kong2023end, wang2025iggm} aim to model the antibody design process holistically. Given a specific antigen and a partial antibody sequence lacking the CDR-H3, these methods directly generate the complete antibody sequence and full-atom 3D structure in a single, unified framework. For example, dyMEAN \citep{kong2023end} performs full-atom generation through modules including antibody initialization, multi-channel equivariant encoding, and joint sequence-structure prediction. Recent methods further realize end-to-end design via flow matching: dyAb \citep{tan2025dyab} generates CDR sequences and structures during a fine-grained optimization stage before grafting them onto a coarse-grained template, while IgFlow \citep{nagaraj2024igflow} performs joint discrete–continuous flow matching on the entire antibody to produce coupled sequence–structure designs.

Despite significant progress in antibody design, two critical challenges persist. \textbf{First, current generative frameworks insufficiently integrate sequence-structure generation with global geometric message passing.} Existing models either perform local (e.g., CDR or paratope) antibody structure generation, but generally fail to model the structural information flow between the paratope and the rest of the antibody or to account for holistic structural refinement, resulting in weak coupling between generative dynamics and structural information propagation \citep{luo2022antigen, zhu2024antibody}. \textbf{This decoupling limits their ability to produce globally coherent full-atom antibodies.}

\textbf{Second, the inability to leverage fine-grained antigen geometry for precise interface modeling.} Current models struggle to utilize detailed geometric features of the antigen surface, limiting their ability to accurately model antigen-antibody interfaces \citep{10.1093/bib/bbae048}. This shortcoming impairs structural signal propagation across the antibody framework, leading to suboptimal interface complementarity and reduced global structural coherence in generated antibodies.

Based on the aforementioned research progress and challenges, we propose \textbf{AbFlow}, a novel one-step flow-based generative framework for end-to-end full-atom antibody design. Our main contributions are as follows:
\setlength{\parskip}{0.05em}
\begin{itemize}\setlength{\itemsep}{0pt}\setlength{\parskip}{0pt}
\item[$\bullet$] Our AbFlow employs a \textbf{paratope-restricted one-step flow matching module} for high-fidelity generation. The framework incorporates an \textbf{EGNN-based joint generative and message-passing velocity field} that performs one-step paratope generation and propagates structural information from the generated paratope to the full antibody. This dual functionality enables coherent full-atom refinement of the entire antibody despite applying flow only to a local region.
\item[$\bullet$] To address the lack of antigen-specific geometric context in existing methods, we propose an \textbf{equivariant Surface Multi-channel Encoder (SME)} that extracts fine-grained structural cues from the antigen surface. This geometric information is subsequently injected into the one-step velocity field network to effectively guide paratope generation in an antigen-aware and interaction-consistent manner.
\item[$\bullet$] Extensive experiments on \textbf{paratope–centric antibody design}, \textbf{multi-CDRs and full-atom antibody design}, \textbf{binding affinity optimization}, and \textbf{complex structure prediction} demonstrate that our AbFlow outperforms existing step-by-step and end-to-end baselines in structural accuracy, interface quality, and binding affinity.
\end{itemize}
See Figure~\ref{fig1} for a schematic comparison of different antibody design paradigms, including step-by-step pipelines, conventional end-to-end frameworks, and our generative approach.
\section{Related Work}
\textbf{Antibody design.} Antibody aims to develop antibodies that can bind to specific antigens with high efficiency. Based on the diffusion probabilistic model, DiffAb \citep{luo2022antigen} iteratively updates the amino acid type, position, and orientation of CDRs to achieve tasks such as sequence-structure co-design. HERN \citep{jin2022antibody} employs a hierarchical equivariant refinement network. The paratope docking task predicts atomic forces through a hierarchical message passing network. MEAN \citep{kong2022conditional} models antibody design as a 3D-equivariant graph translation problem, with the multi-channel equivariant attention network and the progressive full-shot decoding strategy. dyMEAN \citep{kong2023end}, an end-to-end full-atom model, through an adaptive multi-channel encoder, iteratively updates the antibody sequence and structure. Our work further develops this end-to-end method, enhancing the quality of the interaction between the generated antibody and the antigen, as well as their binding affinity.

\textbf{Generative model.} In recent years, generative models, including diffusion-based and flow-based models, have been increasingly used in molecular and protein generation tasks. AlphaFold 3 \citep{abramson2024accurate} replaces the structure module of AlphaFold 2 \citep{jumper2021highly} with a diffusion module for generating 3D structures. DiffAb \citep{luo2022antigen} and AbX \citep{zhu2024antibody} are based on DDPM \citep{ho2020denoising} and score-based models \citep{song2021maximum}, respectively, for antibody CDR design. Flow matching \citep{lipman2022flow, liu2022flowstraightfastlearning, albergo2023buildingnormalizingflowsstochastic} simulates the transport from a simple distribution to the target (data) distribution by learning Continuous Normalizing Flows (CNFs). Flow matching has recently gained widespread use in molecular generation tasks. ETFlow \citep{hassan2024etflowequivariantflowmatchingmolecular} implements molecular 3D conformer optimization based on flow matching, while GOAT \citep{hong2025accelerating3dmoleculegeneration} combines flow and optimal transport theory for 3D molecular generation. In our work, we apply the local one-step flow matching framework to end-to-end antibody design tasks.

\textbf{Surface representation.} The structure surface contains the geometric and chemical information that reflects the interaction patterns and relationship with other molecules. MaSIF \citep{gainza2020deciphering, sverrisson2021fast} decomposes the surface into overlapping radial patches, extracts features, and embeds them into vector descriptors, applying them to various tasks with the help of geometric deep learning. EquiPocket \citep{zhang2023equipocket} treats proteins as graph structures, extracting atomic geometric information and predicting ligand binding sites using an E(3)-equivariant graph neural network for modeling and message passing. SurfPro \citep{song2024surfpro} models surface features with a hierarchical encoder and generates amino acid sequences with an autoregressive decoder. We introduce surface representation into the antibody design task to convey more fine-grained interaction information between the antigen and the antibody.
\section{Notations and Problem Setting}
\label{sec:pro setting}
A general antibody has a characteristic Y-shaped structure composed of two heavy chains and two light chains held together by disulfide bonds. It has constant domains that remain unchanged among different antibodies, and variable regions at the tips of the Y that specifically recognize and bind to antigens. 

Our work lies in the variable regions $V=\{V_h, V_l\}$, where $V_h$ and $V_l$ represent the variable regions in the heavy and light chains. A variable region can be divided into four framework regions (FRs) and three complementarity-determining regions (CDRs). In work related to antibody design, CDR design refers to the design of the coordinates and sequence of the third complementarity-determining region (CDR-H3) of the heavy chain. Since the sequence of this part of the antibody is highly variable and often unknown, it dominates the binding of antigen and antibody \citep{maccallum1996antibody}. In this paper, we refer to CDR-H3 as the \textbf{paratope} according to \citep{jin2022antibody}. Additionally, the part of an antigen that binds to the antibody is called the \textbf{epitope}.

Following \citep{kong2022conditional} and \citep{kong2023end}, we use the graph format $\mathcal{G} = (\mathcal{V},\mathcal{E})$ to represent the epitope and paratope used in this work. Here, node set V consists of all residues in the structure, where $V = \{v_i\}_{i = 1}^{|\mathcal{V}|}$ and $v_i = (a_i, X_i)$. In this, $a_i \in \mathbb{R}$ represents the type of the i-th amino acid, and $X_i \in \mathbb{R}^{3\times c_i}$ represents the full-atom 3D coordinates of the residue, where $c_i$ indicates the number of atoms in this residue. As a full-atom design method, we need to represent the different numbers of atoms in different residues.
The edge set $\mathcal{E}$ is calculated based on the distances between nodes in $\mathcal{G}$. We calculate the pair-wise distance between two residues $v_i$ and $v_j$ as Equaion \eqref{kNN} :
\begin{equation}
    d\left(v_i, v_j\right)=\min _{1 \leq m \leq c_i, 1 \leq n \leq c_j}\left\|\boldsymbol{X}_i(:, m)-\boldsymbol{X}_j(:, n)\right\|_2.
    \label{kNN}
\end{equation}

Here, $X_i(:, m)$ and $X_j(:, n)$, respectively, represent the m-th atom of residue i and the n-th atom of residue j. Then, we calculate the k-nearest neighborhood (KNN) of each residue. In this way, we can define epitope $\mathcal{G}_E=(\mathcal{V}_E, \mathcal{E}_E)$, antibody $\mathcal{G}_A=(\mathcal{V}_A, \mathcal{E}_A)$ and paratope $\mathcal{G}_P=(\mathcal{V}_P, \mathcal{E}_P)$.

\textbf{Problem Setting:} The main task of our work is to design the full-atom antibody 3D structures as well as the paratope sequences under the premise that only the epitope is known and the sequence of the antibody without the CDR-H3 part is available. Expressed in symbols: Given the epitope $\mathcal{G}_E=(\mathcal{V}_E, \mathcal{E}_E)$ and the antibody sequence with missed paratope $\{a_i|i\in \mathcal{V_A}, i \notin \mathcal{V_P}\}$, design paratope sequences $\{a_i|i\in \mathcal V_P\}$ and full-atom antibody 3D coordinates $\{X_i|i\in \mathcal V_A\}$.

Beyond this basic problem setting, AbFlow can be naturally extended to a broader range of antibody modeling tasks. In particular, we demonstrate the scalability of our framework through applications to multi-CDRs design and antigen-antibody complex structure prediction, highlighting its flexibility and generality across different antibody modeling scenarios.
\section{Method}
\label{sec:method}
Figure \ref{fig2} presents an overview of AbFlow. Next, we provide a detailed introduction to each module in the architecture. Section \ref{sec:flow} introduces the one-step flow matching formulation for transporting the initialized paratope state toward the target endpoint. Section \ref{sec:surf} introduces the definition and sampling of the epitope surface. Section \ref{sec:SME} presents the extended one-step flow network, including the Surface Multi-Channel Encoder (SME) used to enhance the interaction between epitope and paratope information, as well as the network that propagates structural information from the paratope to the entire antibody. Section~\ref{sec:train} presents the training objective and training algorithm.
\begin{figure*}[!ht]
\begin{center}
\centerline{\includegraphics[width=1.9\columnwidth]{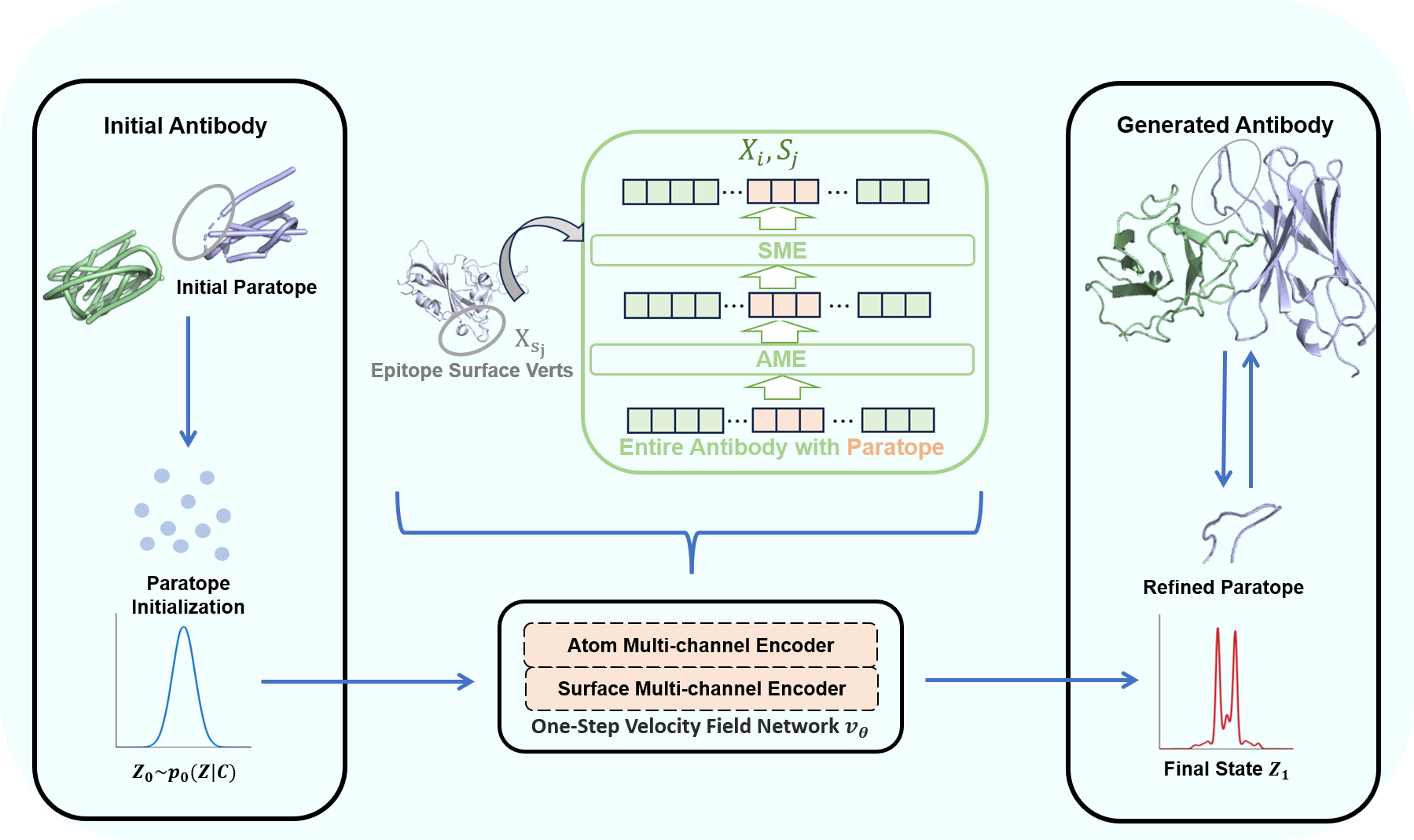}}
\caption{\textbf{Overall architecture of AbFlow.} Input the antigen epitope and the antibody sequence lacking CDR-H3 as introduced in Section \ref{sec:pro setting}. We initialize the local paratope state for the transport of probability distribution. In the single-step forward process, we introduce an extended velocity field network that not only predicts the one-step paratope endpoint but also enhances specific-antigen binding interface information through a Surface Multi-channel Encoder (SME), propagating this information from the paratope to the entire antibody structure.}
\Description{Overall architecture of AbFlow.}
\label{fig2}
\end{center}
\end{figure*}

\subsection{One-Step Flow Matching for the Paratope}
\label{sec:flow}

Flow matching \citep{lipman2022flow} learns a transport process that maps an initial distribution to a target data distribution. Given an initial structural state $X_0 \sim p_0$ and a target structure $X_1 \sim p_1$, standard flow matching models this transport through a time-dependent velocity field $u_{\phi}(X_t,t)$ governed by
\begin{equation}
\frac{\mathrm{d}X_t}{\mathrm{d}t}=u_{\phi}(X_t,t), \qquad X_0\sim p_0.
\label{eq:fm_ode}
\end{equation}

A commonly used conditional probability path is the linear interpolation between the initial and target states. With $\sigma_{\min}=0$, this path and its corresponding target velocity are
\begin{equation}
X_t=(1-t)X_0+tX_1, \qquad u^{*}(X_t,t)=X_1-X_0.
\label{eq:fm_path}
\end{equation}
Standard flow matching trains the velocity field at randomly sampled intermediate states:
\begin{equation}
\mathcal{L}_{\mathrm{FM}}=\mathbb{E}_{t,X_0,X_1}\left[\left\|u_{\phi}(X_t,t)-u^{*}(X_t,t)\right\|_2^2\right], \qquad t\sim\mathcal{U}(0,1).
\label{eq:standard_fm}
\end{equation}

AbFlow adopts a paratope-restricted one-step formulation of this transport process, designed specifically for local interface generation. Let $C$ denote the available antigen-specific conditional information, primarily including the antigen structure and epitope surface geometry, together with the known partial sequence and structural context of the antibody. The initial local interface state is sampled from a conditional initialization distribution:
\begin{equation}
Z_0=(S_0,X_0)\sim p_0(Z\mid C),
\label{eq:conditional_initialization}
\end{equation}
where $S_0$ and $X_0$ denote the initial paratope sequence and coordinates, respectively.

Instead of sampling intermediate flow times, the main model consistently evaluates the transport at the initial boundary:
\begin{equation}
t\sim\delta_0, \qquad X_t=X_0, \qquad u^{*}(X_0,0)=X_1-X_0,
\label{eq:one_step_flow}
\end{equation}
where $\delta_0$ denotes the point-mass distribution at $t=0$. The network therefore learns the complete one-step transport from the initialized local interface to the target paratope structure without requiring explicit continuous-time conditioning.

Let $v_{\theta}=(v_{\theta}^{s},v_{\theta}^{x})$ denote the joint generative and message-passing network. Given the initialized local state and the conditional information, its structural branch $v_{\theta}^{x}$ predicts the refined endpoint coordinates, while its sequence branch $v_{\theta}^{s}$ jointly updates the sequence representation and produces residue-type logits, whose categorical decoding yields the sequence endpoint. The joint prediction is written as
\begin{equation}
(\hat{S}_1,\hat{X}_1)=v_{\theta}(Z_0,C).
\label{eq:one_step_prediction}
\end{equation}
The associated structural displacement $\hat{X}_1-X_0$ represents the predicted one-step transport from the initial state to the target endpoint. Because the transport is evaluated only at the initial boundary, directly predicting $\hat{X}_1$ implicitly determines the complete one-step displacement. Therefore, the adopted one-step flow formulation can also be interpreted as an endpoint mapping from the initialized local state to the refined paratope.

The one-step formulation is applied only to the local paratope region, while the resulting local information is propagated through the entire antibody using the message-passing network described in Section~\ref{sec:SME}. We additionally compare this formulation with standard flow matching trained on randomly sampled intermediate states in Section~\ref{abl}.

\subsection{Surface Definition and Sampling}
\label{sec:surf}
In this section, we describe the procedure for constructing surface vertices around the epitope and sampling representative points for interaction with the paratope. Since the paratope inherently forms a flexible and well-defined surface on the antibody, we focus exclusively on modeling the surface of the epitope to enhance the message passing between the antigen and antibody.

Given the epitope residues $\{X_i \mid i=1, 2, \dots, N=|\mathcal{V}_E|\}$, we utilize MSMS \citep{sanner1996reduced} to generate a dense set of surface vertices $\mathbb{S} = \{s_1, s_2, \dots, s_M\}$ within a radius of 1.5 Å centered on the epitope region, where $M \gg N$. For each vertex $s_i \in \mathbb{S}$, we assign an \emph{attribution residue} by identifying its nearest epitope residue according to the following rule:
\[
v_j = \arg\min_{1 \leq k \leq N} \|s_i - X_k\|_2.
\]
This association partitions the vertex set $\mathbb{S}$ into $N$ disjoint subsets, each corresponding to an epitope residue:
\begin{align}
\mathbb{S}_{v_p} &= \{ s_q \}_{q=1}^{m_p}, \quad p=1,\ldots,N, \\
\sum_{p=1}^{N} m_p &= M, \quad 0 \le m_p \ll M.
\end{align}
To reduce computational cost and maintain robustness across residues, we randomly sample at most $M_0$ surface vertices for each epitope residue when constructing each batch. These sampled vertices are then used to interact with the paratope, injecting local geometric and chemical information from the antigen surface into the antibody representation. Notably, we avoid using all available vertices or a multi-channel aggregation strategy, as the number of vertices associated with each epitope residue varies significantly and can become prohibitively large in practice.

\subsection{SME and One-Step Velocity Field Network}
\label{sec:SME}
Next, to address the challenge that specific antigen structure information is underutilized in guiding antibody design, we propose a Surface Multi-channel Encoder (SME) to propagate finer-grained, surface-level epitope information to paratope structures. SME can uniformly handle residues with different channel sizes on the paratope as well as surface vertices on the epitope surface, whose size is far larger than the channel size of residues. It can also effectively capture the geometric context between surface vertices and residue atoms, enabling fine-grained message construction. By aggregating these geometry-aware messages, the SME module \textbf{facilitates directional information from the epitope surface to the antibody paratope}. The process of a single update can be represented by Equation \eqref{SME1}-\eqref{SME5}. 
\begin{align}
    m_{ij} &= \phi_m\left(S_i^{(l)}, S_j^{(l)}, \text{msg}(X_i^{(l)}, X_{s_j}^{(l)})\right), \label{SME1} \\
    \text{msg}(X_i^{(l)}, X_{s_j}^{(l)})) &= \frac{\phi_v\left(A_i^T D(X_i^{(l)}, X_{s_j}^{(l)})\right)}{\Vert D(X_i^{(l)}, X_{s_j}^{(l)})\Vert_F + \epsilon}, \label{SME2} \\
    X_{ij} &= \sum_{k=1}^M \left(X_i^{(l)} - \mathbf{1}X_{s_j}^{(l)}[k,:]\right), \label{SM3} \\
    S_i^{(l+1)} &= \phi_s(S_i^{(l)}, \sum_{j \in \mathcal{N}_i} m_{ij}), \label{SME4} \\
    X^{(l+1)}_i &= X^{(l)}_i + \frac{1}{|\mathcal{N}_i|} \sum_{j \in \mathcal{N}_i} X_{ij} \phi_x(m_{ij}) \label{SME5}.
\end{align} 
Where $S_i,X_i$ are the sequence and coordinate representation of the i-th residue node, respectively, $X_{s_j}\in \mathbb{R}^{M\times3}$ represents the coordinates of surface vertices belonging to the j-th residue. $\mathcal{N}_i$ represents the neighbor of residue i,  $\phi_m, \phi_v, \phi_x$ are different MLP layers, and $\textbf{1}$ is a vector of all ones. $D(\cdot, \cdot)$ is the calculation of the pair-wise Euclidean distance of two matrices. For the learnable attribute matrix $A_i\in \mathbb{R}^{c_i\times d}$, the vector of each row is learned from the type and position of the corresponding i-th atom. We discuss and prove the E(3)-equivariance of SME in Appendix \ref{E3}.

The SME module is integrated into the one-step velocity field network $v_{\theta}$, enabling it to jointly predict the refined paratope endpoint and propagate antigen-specific messages across the antibody structure. To construct $v_{\theta}$, we combine SME with the Atom Multi-channel Encoder (AME) from \citep{kong2023end}, forming a hybrid message-passing backbone. While SME focuses on fine-grained surface-level transmission of antigen-epitope information to the paratope, AME handles global structural propagation. This design enables one-step paratope generation while retaining full-atom generation capabilities. The joint sequence and structural endpoint prediction produced by $v_{\theta}$ is defined in Equation~\eqref{eq:one_step_prediction}.

\subsection{Loss and Training}
\label{sec:train}

\textbf{Loss.} The model is trained using the sequence loss $\mathcal{L}_{\mathrm{seq}}$, structure loss $\mathcal{L}_{\mathrm{struct}}$, and docking loss $\mathcal{L}_{\mathrm{dock}}$:
\begin{equation}
\mathcal{L}
=
\mathcal{L}_{\mathrm{seq}}
+
\mathcal{L}_{\mathrm{struct}}
+
\mathcal{L}_{\mathrm{dock}}.
\label{LOSS}
\end{equation}
The docking loss consists of the interface endpoint loss $\mathcal{L}_{\mathrm{interface}}$ and the antigen-paratope edge-distance loss $\mathcal{L}_{\mathrm{edge}}$:
\begin{equation}
\mathcal{L}_{\mathrm{dock}}
=
\mathcal{L}_{\mathrm{interface}}
+
\mathcal{L}_{\mathrm{edge}}.
\label{eq:dock_loss}
\end{equation}
The interface loss directly supervises the predicted paratope endpoint:
\begin{equation}
\mathcal{L}_{\mathrm{interface}}
=
\frac{1}{|\mathcal{V}_{P}|}
\sum_{i\in\mathcal{V}_{P}}
\ell_{\mathrm{huber}}
\left(
\hat{X}_{1,i},X_{1,i}
\right),
\label{eq:interface_loss}
\end{equation}
where $\hat{X}_{1,i}$ and $X_{1,i}$ denote the predicted and ground-truth coordinates of the $i$-th paratope residue, respectively. The detailed definitions of $\mathcal{L}_{\mathrm{seq}}$, $\mathcal{L}_{\mathrm{struct}}$, and $\mathcal{L}_{\mathrm{edge}}$ are provided in Appendix~\ref{l&g}.

Under the one-step flow-matching formulation introduced in Section~\ref{sec:flow}, the predicted and target structural displacements from the same local initialization $X_0$ are
\begin{equation}
d_{\theta}
=
\hat{X}_1-X_0,
\qquad
d^{*}
=
X_1-X_0.
\label{eq:one_step_displacements}
\end{equation}
The corresponding theoretical flow loss at $t=0$ can be written as
\begin{equation}
\mathcal{L}_{F}
=
\lambda_F
\frac{1}{|\mathcal{V}_{P}|}
\sum_{i\in\mathcal{V}_{P}}
\ell_{\mathrm{huber}}
\left(
d_{\theta,i},d^{*}_{i}
\right).
\label{eq:one_step_flow_loss}
\end{equation}
Since both displacements share the same initialization, we have
\begin{align}
\mathcal{L}_{F}
&=
\lambda_F
\frac{1}{|\mathcal{V}_{P}|}
\sum_{i\in\mathcal{V}_{P}}
\ell_{\mathrm{huber}}
\left(
\hat{X}_{1,i}-X_{0,i},
X_{1,i}-X_{0,i}
\right)
\nonumber\\
&=
\lambda_F
\frac{1}{|\mathcal{V}_{P}|}
\sum_{i\in\mathcal{V}_{P}}
\ell_{\mathrm{huber}}
\left(
\hat{X}_{1,i},X_{1,i}
\right)
=
\lambda_F\mathcal{L}_{\mathrm{interface}}.
\label{eq:flow_interface_equivalence}
\end{align}
Therefore, under the adopted $t=0$ one-step parameterization, $\mathcal{L}_{F}$ and $\mathcal{L}_{\mathrm{interface}}$ provide equivalent supervision up to the scalar weight $\lambda_F$. Since $\mathcal{L}_{\mathrm{interface}}$ is already included in $\mathcal{L}_{\mathrm{dock}}$, adding $\mathcal{L}_{F}$ separately would only reweight the same endpoint residual without introducing an additional training signal. Consistent with this analysis, explicitly adding the separate flow-loss term did not improve model performance.

Algorithm~\ref{alg1} summarizes the training procedure. Given the antigen-specific conditional information $C$ and a target antibody $(S_1,X_1)$, the model samples an initial local interface state at $t=0$, predicts its sequence and structural endpoints, and optimizes the practical objective in Equation~\eqref{LOSS}.

\begin{algorithm}[t]
\caption{Training Algorithm of AbFlow}
\label{alg1}
\begin{algorithmic}[1]
\State \textbf{Input:} conditional information $C$ and target antibody $(S_1,X_1)$
\State \textbf{Initialize:} one-step velocity field network $v_{\theta}$
\While{$\theta$ has not converged}
\State Sample an initial local interface state $Z_0=(S_0,X_0)\sim p_0(Z\mid C)$ at $t=0$
\State $(\hat{S}_1,\hat{X}_1)\leftarrow v_{\theta}(Z_0,C)$
\State Compute $\mathcal{L}_{\mathrm{seq}}$, $\mathcal{L}_{\mathrm{struct}}$, and $\mathcal{L}_{\mathrm{dock}}$
\State $\mathcal{L}\leftarrow\mathcal{L}_{\mathrm{seq}}+\mathcal{L}_{\mathrm{struct}}+\mathcal{L}_{\mathrm{dock}}$
\State Update $\theta$ using $\nabla_{\theta}\mathcal{L}$
\EndWhile
\State \Return $v_{\theta}$
\end{algorithmic}
\end{algorithm}

\textbf{Sampling.} Starting from an initialized local interface state $Z_0=(S_0,X_0)\sim p_0(Z\mid C)$ at $t=0$, AbFlow performs one-step endpoint generation. The structural branch directly predicts the local structural endpoint $\hat{X}_1$, while the sequence branch produces residue-type logits that are converted into a categorical distribution $\hat{q}_1$. The residue types are decoded from this distribution to obtain the sequence endpoint $\hat{S}_1$. The complete sampling and generation procedure is presented in Appendix~\ref{sec:generation}.

During training, the one-step velocity field network learns the transport from the conditionally initialized paratope state to its target endpoint. During sampling, the network performs one-step endpoint generation from a local interface state initialized at $t=0$. The resulting paratope information is propagated across the antibody through the global message-passing component of $v_{\theta}$. 
\section{Experiments}
\label{experiments}
\begin{table*}[ht]
    \centering
    \caption{Comparison of Different Models on the Paratope-centric Antibody Design Task. AbFlow results include standard deviations and their corresponding $2\sigma$ confidence coverage to reflect statistical stability.}
    \label{exp1}
    \begin{tabular}{c|ccc|ccc}
        \toprule
        \multirow{2}{*}{Model} & \multicolumn{3}{c}{Overall} & \multicolumn{3}{c}{Interface} \\
        \cline{2 - 4} \cline{5 - 7} 
        & AAR$\uparrow$ & TMscore$\uparrow$ & IDDT$\uparrow$ & CAAR$\uparrow$ & RMSD$\downarrow$ & DockQ$\uparrow$ \\
        \midrule
        RosettaAb \citep{adolf2018rosettaantibodydesign} & 0.3231 & 0.9717 & 0.8272 & 0.1458 & 17.70 & 0.137 \\
        DiffAb \citep{luo2022antigen} & 0.3531 & 0.9695 & 0.8281 & 0.2217 & 23.24 & 0.158 \\
        MEAN \citep{kong2022conditional} & 0.3738 & 0.9688 & 0.8252 & 0.2411 & 17.30 & 0.162 \\
        HERN \citep{jin2022antibody} & 0.3265 & - & - & 0.1927 & 9.15 & 0.294 \\
        dyMEAN \citep{kong2023end} & \textbf{0.4365} & \underline{0.9726} & \underline{0.8454} & \underline{0.2811} & \textbf{8.11} & \underline{0.409} \\
        dyAb \citep{tan2025dyab} & 0.3789 & 0.9264 & 0.6957 & 0.2614 & 9.86 & 0.342 \\
        \hline
        \textbf{AbFlow} & \underline{0.4234}{\footnotesize$\pm$0.1334 (96.7\%)} & \textbf{0.9736}{\footnotesize$\pm$0.009 (93.3\%)} & \textbf{0.8522}{\footnotesize$\pm$0.0216 (95\%)} & \textbf{0.2824}{\footnotesize$\pm$0.2037 (96.7\%)} & \underline{8.25}{\footnotesize$\pm$4.55 (98.3\%)} & \textbf{0.423}{\footnotesize$\pm$0.161 (93.3\%)} \\
        \textbf{AbFlow-G} & 0.4204 & 0.9736 & 0.8522 & 0.279 & 8.28 & 0.422 \\
        \bottomrule
    \end{tabular}
\end{table*}
To validate the effectiveness of our proposed AbFlow, we conducted experiments on 4 critical tasks related to antibody generation: 1) paratope-centric antibody design (\ref{sec:exp1}), multi-CDRs and full-atom antibody design (\ref{sec:exp2}), binding affinity optimization (\ref{sec:exp3}), and complex structure prediction (\ref{sec:exp4}). We also provide details of our experimental setups (in \ref{sec:baslines and metrics}) and ablation studies on key modules (in Section \ref{abl}) to verify their contributions.
\subsection{Experiment Setup}
\label{sec:baslines and metrics}
\textbf{Datasets}. The experiments are trained on the Structural Antibody Database (SAbDab) \citep{dunbar2014sabdab}. For the full-atom antibody design task, the RAbD benchmark \citep{adolf2018rosettaantibodydesign} composed of 60 diverse complexes selected by domain experts is used for test. For the affinity optimization task, antibodies from SKEMPI V2.0 \citep{jankauskaite2019skempi} are used for evaluation. We confirm that there are no duplicate samples between the training set and the test set to avoid data leakage. More details about the dataset are presented in Appendix \ref{data}.

\textbf{Baselines}. For end-to-end full-atom antibody design methods, we use three representative baselines spanning different generative paradigms: dyMEAN \citep{kong2023end} (GNN-based), IgGM \citep{wang2025iggm} (diffusion-based), and dyAb \citep{tan2025dyab} (flow-based). For other models used in CDR-H3 generation, including RosettaAb \citep{adolf2018rosettaantibodydesign}, MEAN \citep{kong2022conditional}, DiffAb \citep{luo2022antigen}, and HERN \citep{jin2022antibody}, we standardized the remaining steps of the step-by-step pipeline to ensure a fair comparison with our AbFlow. Specifically, we use IgFold \citep{ruffolo2023fast} for antibody structure prediction, HDock \citep{yan2020hdock} for antigen-antibody docking, and Rosetta \citep{alford2017rosetta} is used for side-chain packing. Here we use IgFold instead of AlphaFold \citep{jumper2021highly} due to its specialization in antibody structure prediction.

\textbf{Metrics}. For the full-atom antibody design task, we employed metrics evaluating sequence recovery (AAR, CAAR \citep{ramaraj2012antigen}) and structural similarity (TM-score \citep{zhang2007scoring, xu2010significant}, lDDT \citep{mariani2013lddt}, RMSD \citep{kabsch1976solution}, DockQ \citep{basu2016dockq, mirabello2024dockq}) to assess the quality of generated antibody structures in terms of both structural fidelity and antigen-antibody docking interface geometry. For the binding affinity optimization task, we utilized binding free energy change ($\Delta\Delta G$) \citep{shan2022deep} and IMP to quantify optimization efficacy and success rates. Detailed definitions of these evaluation metrics are provided in Section \ref{sec:exp2} and Appendix \ref{metrics}.
\subsection{Paratope-centric Antibody Design}
\label{sec:exp1}

This task is the main experiment outlined in Section \ref{sec:pro setting}, which involves the co-design of the 1D CDR-H3 sequence and the 3D coordinates of the full-atom antibody. We compare our AbFlow model with several baseline methods. As shown in Table \ref{exp1}, AbFlow achieves the best performance across most evaluation metrics. This includes both the entire antibody structure (TMscore, lDDT) and the CDR-H3 sequence along with its interface with the epitope (CAAR, DockQ). Notably, the RMSD reported here corresponds to the entire antigen-antibody complex after structural alignment of the CDR-H3 region. This alignment is performed to focus on the quality of the entire antibody structure, with particular attention to the CDR-H3 region, which plays a central role in antigen binding. Additionally, we also evaluated AbFlow on the updated DockQ v2 \citep{mirabello2024dockq} across the entire test set, where AbFlow scored 0.339 compared to dyMEAN's 0.332, further demonstrating AbFlow's superior performance.

To assess the stability of our model, we also provide results from multiple sampling called AbFlow-G(Generator), where 30 samples were generated for each test target. The results demonstrate high consistency across repeated trials, highlighting the robustness of the model, which is enabled by the direct one-step flow mechanism that supports stable endpoint generation and efficient sampling.

Figure \ref{fig3} presents a case study from our full-atom antibody design task. In this case, the SME module enhances the interaction between the epitope and the paratope. Compared to dyMEAN, AbFlow better avoids steric clashes in the generated structure, which is crucial for improving the rationality and stability of the antibody design. Reducing steric clashes—defined as the undesirable overlap between atoms—leads to more physically plausible structures and improves the quality of the antigen-antibody interface, which better reflects the true nature of their interaction.

\begin{figure}[ht]
\begin{center}
\centerline{\includegraphics[width=1.0\columnwidth]{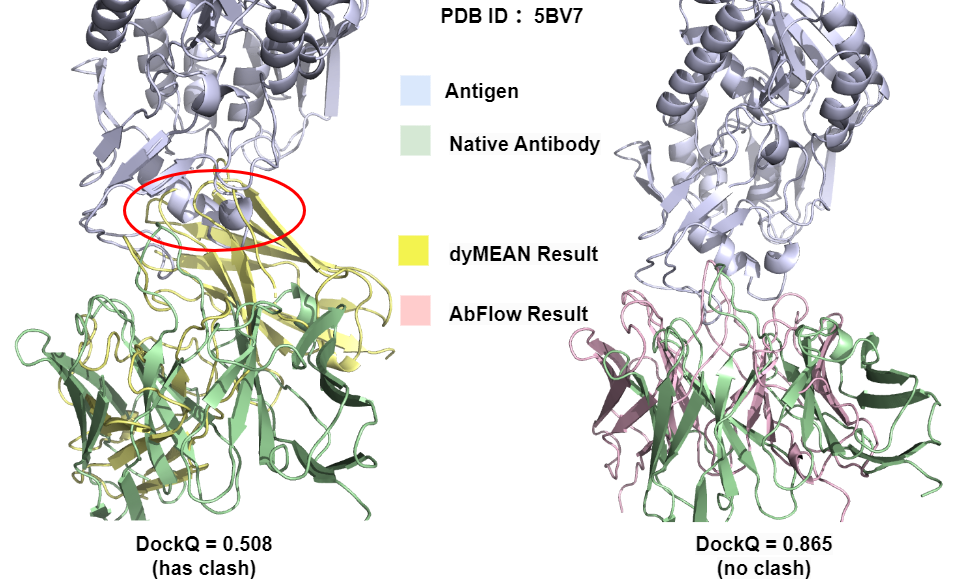}}
\caption{\textbf{Case Study of the Full-atom Antibody Design Task.} The structures on the left and right show antibodies generated by dyMEAN and AbFlow, respectively. The antibody designed by dyMEAN (yellow) exhibits a distinct steric clash (red circle) with the epitope, while AbFlow (pink) effectively avoids this clash, generating a structure highly similar to the true antibody (green). }
\Description{Case Study of the Full-atom Antibody Design Task.}
\label{fig3}
\end{center}
\end{figure}

\subsection{Multi-CDRs and Full-atom Antibody Design}
\label{sec:exp2}
To evaluate the scalability of AbFlow beyond CDR-H3 generation, we further conduct multi-CDR design experiments on the RAbD test set, covering all six CDRs (H1–H3, L1–L3). As shown in Table \ref{exp4}, AbFlow consistently achieves the lowest RMSD across all CDRs, demonstrating its ability to model cross-CDR structural dependencies and maintain global structural coherence even when multiple loops are generated simultaneously. Our AbFlow outperforms DiffAb \citep{luo2022antigen}, dyMEAN \citep{kong2023end}, and AbX \citep{zhu2024antibody} by a clear margin on both heavy-chain and light-chain loops. These results confirm that the paratope-restricted one-step flow matching and EGNN-based refinement in AbFlow generalize effectively to multi-region antibody design.

Beyond CDR-level evaluation, we further assess AbFlow on full-atom antibody structure prediction, where only dyMEAN provides comparable baselines among existing end-to-end methods. As presented in Table~\ref{exp4}, AbFlow achieves superior performance on all global structural metrics, including RMSD, TMscore, lDDT, and DockQ, highlighting its advantage in generating accurate and coherent full-antibody structures.

\begin{table}[ht]
\centering
\caption{Multi-CDRs and full-atom design results.}
\label{exp4}
\begin{tabular}{ccccc}
\hline
\rowcolor{gray!20}
\multicolumn{5}{c}{\textbf{Multi-CDRs Results}} \\
\hline
Method & CDR (H) & RMSD $\downarrow$ & CDR (L) & RMSD $\downarrow$ \\
\hline
DiffAb & \multirow{4}{*}{H1} & 0.88 & \multirow{4}{*}{L1} & 0.85 \\
dyMEAN &                     & 1.09 &                     & 1.03 \\
AbX    &                     & 0.85 &                     & 0.78 \\
\textbf{AbFlow} &            & \textbf{0.63} & & \textbf{0.64} \\
\hline
DiffAb & \multirow{4}{*}{H2} & 0.78 & \multirow{4}{*}{L2} & 0.55 \\
dyMEAN &                     & 1.11 &                     & 0.66 \\
AbX    &                     & 0.76 &                     & 0.45 \\
\textbf{AbFlow} &            & \textbf{0.55} & & \textbf{0.25} \\
\hline
DiffAb & \multirow{4}{*}{H3} & 2.86 & \multirow{4}{*}{L3} & 1.39 \\
dyMEAN &                     & 3.88 &                     & 1.44 \\
AbX    &                     & 2.50 &                     & 1.18 \\
\textbf{AbFlow} &            & \textbf{1.83} & & \textbf{0.65} \\
\hline
\rowcolor{gray!20}
\multicolumn{5}{c}{\textbf{Full-atom Results}} \\
\hline
 & RMSD $\downarrow$ & TMscore $\uparrow$ & lDDT $\uparrow$ & DockQ $\uparrow$\\
\hline
dyMEAN & 1.357 & 0.9653 & 0.8028 & 0.396 \\
\textbf{AbFlow} & \textbf{1.102} & \textbf{0.9707} &
\textbf{0.8151} & \textbf{0.431} \\
\hline
\end{tabular}
\end{table}

\subsection{Binding Affinity Optimization}
\label{sec:exp3}
Antibody binding affinity optimization is essential for enhancing the interaction between antibodies and their target antigens, improving specificity and efficacy while minimizing off-target effects. To achieve this, it is crucial to optimize the binding affinity by altering as few residues as possible. In this task, we use $\Delta\Delta G$ as a key metric to evaluate the optimization effect, which represents the change in binding affinity, and $\Delta L$, which tracks the number of residues altered after optimization. These $\Delta\Delta G$ values are predicted using a GNN-based model, which is further refined by training a multi-layer perceptron (MLP) to improve accuracy. Additionally, we introduce the IMP (IMprovement Percentage), which quantifies the proportion of improvement in binding affinity achieved after optimization.

\begin{table}[ht]
    \centering
    \caption{Affinity Optimization Comparison of Different Models.}
    \label{exp2}
    \begin{tabular}{c|ccc}
    \hline
    & \thead{Best $\Delta\Delta G\downarrow$} & \thead{$\Delta L\downarrow$} & $IMP(\%) \uparrow$\\
    \hline
    DiffAb \citep{luo2022antigen} & -2.17 & 7.06 & - \\
    MEAN \citep{kong2022conditional} & -6.48 & 8.96 & - \\
    dyMEAN \citep{kong2023end} & \underline{-7.7} & \textbf{4.92} & 36.7\% \\
    \hline
    \textbf{AbFlow} & \textbf{-16.46} & \underline{6.79} & \textbf{38.4\%}\\
    \hline
    \end{tabular}
\end{table}

Table \ref{exp2} presents a comparison of our AbFlow model with other methods, including DiffAb \citep{luo2022antigen}, MEAN \citep{kong2022conditional}, and dyMEAN \citep{kong2023end}, across three key metrics: best $\Delta\Delta G$, $\Delta L$, and IMP. The results show that AbFlow outperforms the other methods in terms of both reducing the changes in binding affinity and minimizing the number of altered residues. Furthermore, AbFlow achieves a higher IMP, indicating a greater proportion of improvement in binding affinity. These results highlight that AbFlow is not only effective in optimizing the binding affinity but also makes acceptable structural changes to the antibody. 

\begin{figure}[ht]
\begin{center}
\centerline{\includegraphics[width=0.9\columnwidth]{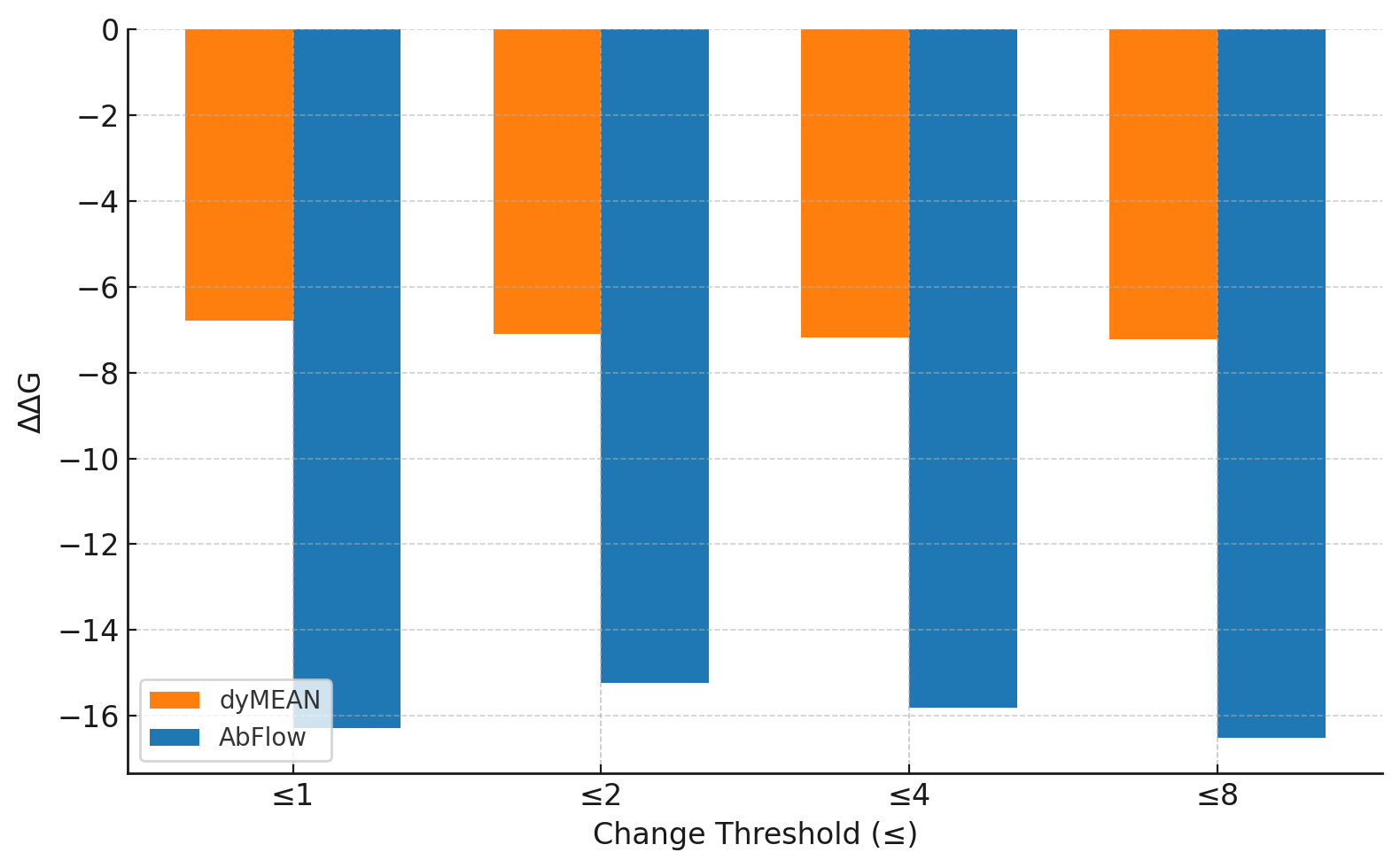}}
\caption{\textbf{$\Delta\Delta G$ Comparison at Different Length of Changed Residues.}}
\Description{Different Length of Changed Residues.}
\label{fig4}
\end{center}
\end{figure}
Figure \ref{fig4} presents a comparison of affinity optimization results when the number of changed residues is constrained. Specifically, the optimization was performed under various thresholds, limiting the maximum number of residues that could be altered to 1, 2, 4, and 8. The figure compares the performance of AbFlow with dyMEAN across these different thresholds. It shows that AbFlow maintains stable optimization performance, consistently achieving favorable $\Delta\Delta G$ values, even when the number of residues altered is limited.

\subsection{Complex Structure Prediction}
\label{sec:exp4}

We also evaluated the performance of AbFlow on the complex structure prediction task, specifically in predicting the antigen-antibody complex structure given the complete antibody sequences. The results were compared to those of HDock \citep{yan2020hdock}, HERN \citep{jin2022antibody}, dyMEAN \citep{kong2023end}, and AlphaFold 3 \citep{abramson2024accurate}. Notably, both HDock and HERN are docking methods, requiring first the prediction of the antibody structure using IgFold \citep{ruffolo2023fast}, followed by docking the antigen and antibody using their respective methods, and then assembling the side-chain with Rosetta. In contrast, AlphaFold 3 directly predicts the complex structure when provided with the sequences of both the antigen and the antibody.

\begin{table}[ht]
    \centering
    \caption{Antigen-antibody Structure Prediction Comparison of Different Models.}
    \label{exp3}
    \begin{tabular}{c|cccc}
    \hline
    Model & TMscore$\uparrow$ & lDDT$\uparrow$ & RMSD $\downarrow$ & DockQ$\uparrow$ \\
    \hline
    HDock \citep{yan2020hdock} & 0.9701 & 0.8439 & 16.32 & 0.202 \\
    HERN \citep{jin2022antibody} & 0.9702 & 0.8441 & 9.63 & 0.429 \\
    dyMEAN \citep{kong2023end} & \textbf{0.9731} & \textbf{0.8673} & \underline{9.05} & \underline{0.452} \\
    \hline
    AlphaFold 3 \citep{abramson2024accurate}& \textbf{0.9799} & \textbf{0.921} & 9.47 & \textbf{0.475} \\
    \textbf{AbFlow} & \underline{0.9714} & \underline{0.8618} & \textbf{9.01} & \textbf{0.455} \\
    \hline
    \end{tabular}
\end{table}

Interestingly, our tests also reveal that AbFlow performs better on this task using its dedicated structure-prediction inference path, where the joint network directly predicts the complex structure. This task-specific configuration improves the performance of end-to-end protein structure prediction.

We evaluated the methods using TM-score, lDDT, and DockQ, as shown in Table \ref{exp3}. Our results indicate that AbFlow outperforms all methods except AlphaFold3 in terms of docking metrics (RMSD DockQ), which reflects the quality of interaction prediction. While AlphaFold3 excels at predicting monomer structures, it leaves room for improvement in protein complex predictions. The gap between our method and AlphaFold3, in this regard, is not insurmountable, providing a clear direction for future development.

\section{Ablation Study}
\label{abl}
In this section, we evaluate the effects of the proposed framework and modules in AbFlow.

\begin{table*}[t]
    \centering
    \caption{Ablation Study Results on the Full-atom Antibody Design}
    \label{abl1}
    \begin{tabular}{c|cccccc}
    \hline
     & AAR$\uparrow$ & TMscore$\uparrow$ & lDDT$\uparrow$ & CAAR$\uparrow$ & RMSD$\downarrow$ & DockQ$\uparrow$ \\
    \hline
    \textbf{AbFlow} & \textbf{0.4234} & \textbf{0.9736} & \textbf{0.8522} & \textbf{0.2824} & 8.25 & \textbf{0.423} \\
    w/o One-step Flow & 0.4000 & 0.9721 & 0.8448 & 0.2640 & \textbf{7.76} & 0.417 \\
    w/o SME & 0.4046 & 0.9732 & 0.8506 & 0.2602 & 8.12 & 0.335 \\
    Standard FM & 0.2789 & 0.9649 & 0.8256 & 0.1382 & 17.11 & 0.243 \\
    \hline
    \end{tabular}
\end{table*}

Table~\ref{abl1} presents the ablation results for the full-atom antibody design task. The results demonstrate that both the one-step flow formulation and SME improve full-atom antibody generation. The \textbf{w/o One-step Flow} setting removes the one-step flow formulation and reduces the model to a direct end-to-end predictor, with the joint network serving as the backbone. This comparison indicates that the one-step flow formulation improves the overall accuracy of sequence-structure co-design. The \textbf{w/o SME} setting shows that SME notably improves interface quality. The \textbf{Standard FM} setting replaces the one-step formulation with a conventional multi-step flow-matching process. The performance degradation suggests that the proposed one-step transport is better suited to paratope-centric sequence-structure co-design than iterative multi-step flow integration.


\begin{table}[ht]
    \centering
    \caption{Ablation Study Results on the Affinity Optimization}
    \label{abl2}
    \begin{tabular}{c|ccc}
    \hline
    & \thead{Best $\Delta\Delta G\downarrow$} & \thead{$\Delta L \downarrow$} & \thead{$IMP\% \uparrow$}\\
    \hline
    \textbf{AbFlow} & \textbf{-16.46} & 6.79 & 38.4\% \\
    w\verb|/|o One-step Flow & -5.41 & \textbf{2.59}  & \underline{40\%} \\
    w\verb|/|o SME & \underline{-5.71} & \underline{4.41} & \textbf{40.8\%} \\
    dyMEAN & -7.7 & 4.92  & 36.7\% \\
    \hline
    \end{tabular}
\end{table}

Table \ref{abl2} presents the ablation study results of AbFlow on the binding affinity optimization task. The results show that the complete AbFlow model significantly improves the affinity optimization performance. Even though it tends to change more residues, the resulting improvement in $\Delta\Delta G$ is substantial. We believe that the optimization gains achieved at the cost of $\Delta L$ are acceptable.

\section{Computational Efficiency}
To assess the computational efficiency of AbFlow, we compare its empirical runtime against predictive end-to-end methods and step-by-step generative baselines.

Compared with dyMEAN, AbFlow performs one-step local transport together with surface-enhanced global message passing. This introduces additional computation associated with the sampled surface vertices $M_0$ and the joint processing of the local interface and full antibody. Despite this increase in computational cost, as shown in Table~\ref{tab:efficiency}, AbFlow achieves substantially stronger structural performance, representing a favorable trade-off relative to dyMEAN's simpler predictive formulation.

Compared with step-by-step methods such as DiffAb, AbFlow demonstrates advantages in both accuracy and efficiency. DiffAb requires $13\,\mathrm{s}$ for CDR-H3 diffusion and $9.89\,\mathrm{s}$ for template generation, resulting in $22.89\,\mathrm{s}$ in total, whereas AbFlow completes full-atom generation in $14\,\mathrm{s}$ while producing significantly better structures. These results demonstrate that coupling paratope-restricted one-step transport with surface-enhanced global message passing provides an efficient and effective framework for full-atom antibody design.

\begin{table}[ht]
\centering
\caption{Comparison of training and inference efficiency per sample.}
\label{tab:efficiency}
\begin{tabular}{lcc}
\hline
\textbf{Model} & \textbf{Train (s)} & \textbf{Inference (s)} \\
\hline
AbFlow & 4.56 & 14.00 \\
dyMEAN & 2.21 & 0.47 \\
DiffAb & -- & 22.89 \\
\hline
\end{tabular}
\end{table}

\section{Future Work}
\label{sec:lim}
While our method demonstrates promising results, several limitations remain to be addressed. The surface information aggregation mechanism, though effective for capturing geometric and physicochemical features, introduces computational overhead that prolongs inference time and may limit real-time or high-throughput applications. A promising direction for future work is to develop adaptive surface sampling strategies—dynamically selecting an appropriate number of surface vertices for each epitope based on its geometry—to further reduce computational cost while preserving the fidelity of geometric information.

On the application side, our problem setting assumes CDR-dominated binding, which limits applicability to cases where framework–epitope interactions are essential. Extending AbFlow to handle such framework-dominated regimes represents an important direction for future work. In addition, although our study focuses on antibodies, extending the framework to broader protein complexes (e.g., protein–ligand or peptide binders) remains an open challenge. Such generalization would further expand the applicability of our approach to wider protein engineering tasks.
\section{Conclusion}
\label{con}
In this study, we introduced AbFlow, an end-to-end one-step flow-matching framework designed to optimize the full-atom antibody design process. By leveraging paratope-centric distribution transport and enhancing the velocity field network with an equivariant Surface Multi-channel Encoder (SME), AbFlow offers efficient full-atom antibody design. Our experiments demonstrate that AbFlow significantly outperforms existing generative and predictive methods in generating structurally accurate antibodies with improved binding affinity. Furthermore, AbFlow is poised to accelerate the development of therapeutic antibodies, especially in scenarios requiring rapid responses to emerging diseases or the creation of personalized immunotherapies.
\begin{acks}
This research was supported in part by National Natural Science Foundation of China (No. 92470128, No. U2241212, No.62376276), by Beijing Nova Program (No. 20230484278). We also wish to acknowledge the support provided by the fund for building world-class universities (disciplines) of Renmin University of China, by Intelligent Social Governance Interdisciplinary Platform, Major Innovation \& Planning Interdisciplinary Platform for the “Double-First Class” Initiative, Public Policy and Decision-making Research Lab, and Public Computing Cloud, Renmin University of China.
\end{acks}
\bibliographystyle{ACM-Reference-Format}
\bibliography{reference}

\appendix
\section*{Appendix}

\section{E(3)-Equivariance of the velocity field network}
\label{E3}
In this section, we prove the E(3)-equivariance of our proposed velocity field network. Note that \citep{kong2023end} has already demonstrated the E(3)-equivariance of the AME module; thus, it suffices to prove the equivariance of the SME module through the following theorem.

\begin{theorem}
Let $\{X_i^{(l)}\}_{i \in \mathcal{V}_A\cup\mathcal{V}_E}$ denote the coordinates of residue nodes and $\{X_{s_j}^{(l)}\}_{j \in \mathcal{N}_i}$ the coordinates of interacting surface points. Let $\{S_i^{(l)}\}_{i \in \mathcal{V}_A\cup\mathcal{V}_E}$ denote the hidden states of residue sequence, respectively. Then, under any Euclidean transformation $g = (Q, t) \in \mathrm{E}(3)$, the Surface Multi-channel Encoder is E(3)-invariant with respect to sequence inputs and E(3)-equivariant with respect to coordinate inputs. Express this in formulation: 
\begin{equation}
    S_i^{(l)}\mapsto S_i^{(l)}
\quad\Longrightarrow\quad
S_i^{(l+1)}\mapsto S_i^{(l+1)}.
\end{equation}
\begin{equation}
    X_i^{(l)} \mapsto QX_i^{(l)} + t \quad \Rightarrow \quad X_i^{(l+1)} \mapsto QX_i^{(l+1)} + t.
\end{equation}
\end{theorem}
\textbf{Proof. }First, it is evident that the surface vertices corresponding to a residue are determined by their relative positional relationships, and thus remain invariant under the transformation $g \in \mathrm{E}(3)$. Meanwhile, the coordinates of surface vertices $\{X_{s_j}^{(l)}\}_{j \in \mathcal{N}_i}$ vary synchronously with the residues coordinates $X_i$. This implies:
\begin{equation}
    X_i \mapsto QX_i + t \Rightarrow X_{s_j} \mapsto QX_{s_j} + t, \forall j\in \mathcal{N}_i.
\end{equation}

For a residue node $i$ and a neighboring surface node $s_j$, the $m_{ij}$ is computed via
\begin{equation}
    m_{ij} = \phi_m(S_i^{(l)}, S_{s_j}^{(l)}, \text{msg}(X_i^{(l)}, X_{s_j}^{(l)})).
\end{equation}

Here, the message is constructed as:
\begin{align}
    \text{msg}(X_i^{(l)}, X_{s_j}^{(l)}) 
    &= \phi_v\left( 
        \frac{A_i^\top D(X_i^{(l)}, X_{s_j}^{(l)})}
             {D(X_i^{(l)}, X_{s_j}^{(l)}) + \epsilon} 
    \right), \nonumber \\
    D[p, q] &= \|X_i^{(l)}[p,:] - X_{s_j}^{(l)}[q, :]\|_2.
\end{align}

for each element in matrix $D$ that indicates Euclidean distance between a residue and a surface vert, satisfies:
\begin{align*}
\| (Q X_i[p, :] + t) - (Q X_{s_j}[q, :] + t) \|_2 = \| Q (X_i[p, :] - X_{s_j}[q, :]) \|_2 \\
= \sqrt{[Q (X_i[p, :] - X_{s_j}[q, :])]^\top [Q (X_i[p, :] - X_{s_j}[q, :])]} \\
= \sqrt{(X_i[p, :] - X_{s_j}[q, :])^\top Q^\top Q (X_i[p, :] - X_{s_j}[q, :])} \\
= \sqrt{(X_i[p, :] - X_{s_j}[q, :])^\top (X_i[p, :] - X_{s_j}[q, :])} \\
= \| X_i[p, :] - X_{s_j}[q, :] \|_2.
\end{align*}
So we have $ \text{msg}(QX_i^{(l)} + t, QX_{s_j}^{(l)} + t) = \text{msg}(X_i^{(l)}, X_{s_j}^{(l)}) $, thus $m_{ij}$ is invariant under $g= (Q, t)$. According to the update process of $S_i$: 
\begin{equation}
    S_i^{(l+1)} = \phi_s(S_i^{(l)}, \sum_{j \in \mathcal{N}_i} m_{ij}).
\end{equation}

we can prove that our SME module remains invariant with respect to the sequence hidden states $S_i$.

Next, applying transformation $g$ to the geometric messages $X_{ij}$ involved in the aggregation process, we have:
\begin{align}
    X_i^{(l)} - \mathbf{1}X_{s_j}^{(l)}[k, :] 
    \mapsto\ & QX_i^{(l)} + t - (Q\mathbf{1}X_{s_j}^{(l)}[k, :] + t) \nonumber \\
    &= Q(X_i^{(l)} - \mathbf{1}X_{s_j}^{(l)}[k, :]), \quad k=1, \dots, M.
\end{align}

thus $X_{ij} \mapsto QX_{ij}$.

SME updates the residues coordinates by:
\begin{equation}
    X_i^{(l+1)} = X_i^{(l)} + \frac{1}{|\mathcal{N}_i|} \sum_{j \in \mathcal{N}_i} X_{ij} \cdot \phi_e(m_{ij}).
\end{equation}

As $X_i^{(l)} \mapsto QX_i^{(l)} + t$, $X_{ij} \mapsto QX_{ij}$, and $\phi_e(m_{ij})$ is invariant, we obtain:
\begin{equation}
    X_i^{(l+1)} \mapsto QX_i^{(l)} + t + \frac{1}{|\mathcal{N}_i|} \sum_j QX_{ij} \cdot \phi_e(m_{ij}) = QX_i^{(l+1)} + t.
\end{equation}

Thus far, we have rigorously proven that the Surface Multi-channel Encoder (SME) module exhibits E(3)-equivariance with respect to the input residue coordinates.

\section{Data and Model Availability}
\label{data}
All datasets and model weights are publicly available in our \href{https://huggingface.co/wenda8759/AbFlow}{Hugging Face repository}, including PDB files with train/validation/test splits and trained model weights for all experiments.

\section{Evaluation Metrics}
\label{metrics}
We use the following several metrics to compare the performance of our AbFlow with other models on full-atom antibody design task: 

\textbf{Amino Acid Recovery (AAR)}: The overlapping ratio of the generated sequence and the ground truth. 

\textbf{CAAR} \citep{ramaraj2012antigen}: Computes AAR restricted to binding residues whose minimum distance from epitope residues is below 6.6Å. 

\textbf{TMscore} \citep{zhang2007scoring, xu2010significant}: Measures the global similarity between the generated structure and the ground truth in terms of $C_{\alpha}$ coordinates. 

\textbf{Local Distance Difference Test (lDDT)} \citep{mariani2013lddt}: Contrasts the difference of the atom-wise distance between the generated structure and the ground truth to evaluate local similarity. 

\textbf{RMSD}: Calculates the Root Mean Square Deviation regarding the absolute coordinates of CDR-H3 without Kabsch \citep{kabsch1976solution, mirabello2024dockq} alignment. 

\textbf{DockQ}: \citep{basu2016dockq}: A comprehensive score for the docking quality of a complex structure.

$\mathbf{\Delta\Delta G}$: Represents the change in Gibbs free energy of the protein complex upon mutation, reflecting the effect of the mutation on protein stability or binding affinity.

For the metrics used to evaluate binding affinity optimization in Section \ref{sec:exp2}: 
\begin{align}
    best \Delta\Delta G=\frac{1}{D}\sum_{d=1}^D \mathop{\min}_{p} \Delta\Delta G_{dp}, \\
    mean \Delta\Delta G=\frac{1}{DP}\sum_{d=1}^D\sum_{p=1}^P \Delta\Delta G_{dp}.
\end{align}

\section{Local Paratope Initialization}
\label{sec:pep_condition}

For settings where an additional paratope prior is available, we employ the receptor-conditioned peptide generation method PepGLAD~\citep{kong2024full} to generate an initial paratope sequence and structure around the epitope. The generated peptide is incorporated into the paratope region of the global antibody context, while the known non-paratope regions remain unchanged. This provides an epitope-conditioned paratope prior for subsequent message passing rather than directly determining the generated interface. The local interface state used by the one-step velocity field network is initialized independently and subsequently refined toward the generated endpoint.

\section{Settings of Surface Sampling}
\label{ss}
Here, we describe how the number of surface vertices on the epitope was determined. Based on the surface cutoff of 1.5~\AA introduced in Section~\ref{sec:surf}, we computed the number of surface vertices per residue across the training set and obtained an average of approximately $\overline{|\mathbb{S}|}\approx 75$. Since each residue is assigned a fixed-size surface vertex set $M_0$, residues with fewer than $M_0$ vertices are padded with zero vectors. To ensure that most residues provide sufficient geometric information without excessive padding, we set $M_0=\frac{2}{3}\,\overline{|\mathbb{S}|}\approx 50$.
\begin{figure}[ht]
\begin{center}
\centerline{\includegraphics[width=0.8\columnwidth]{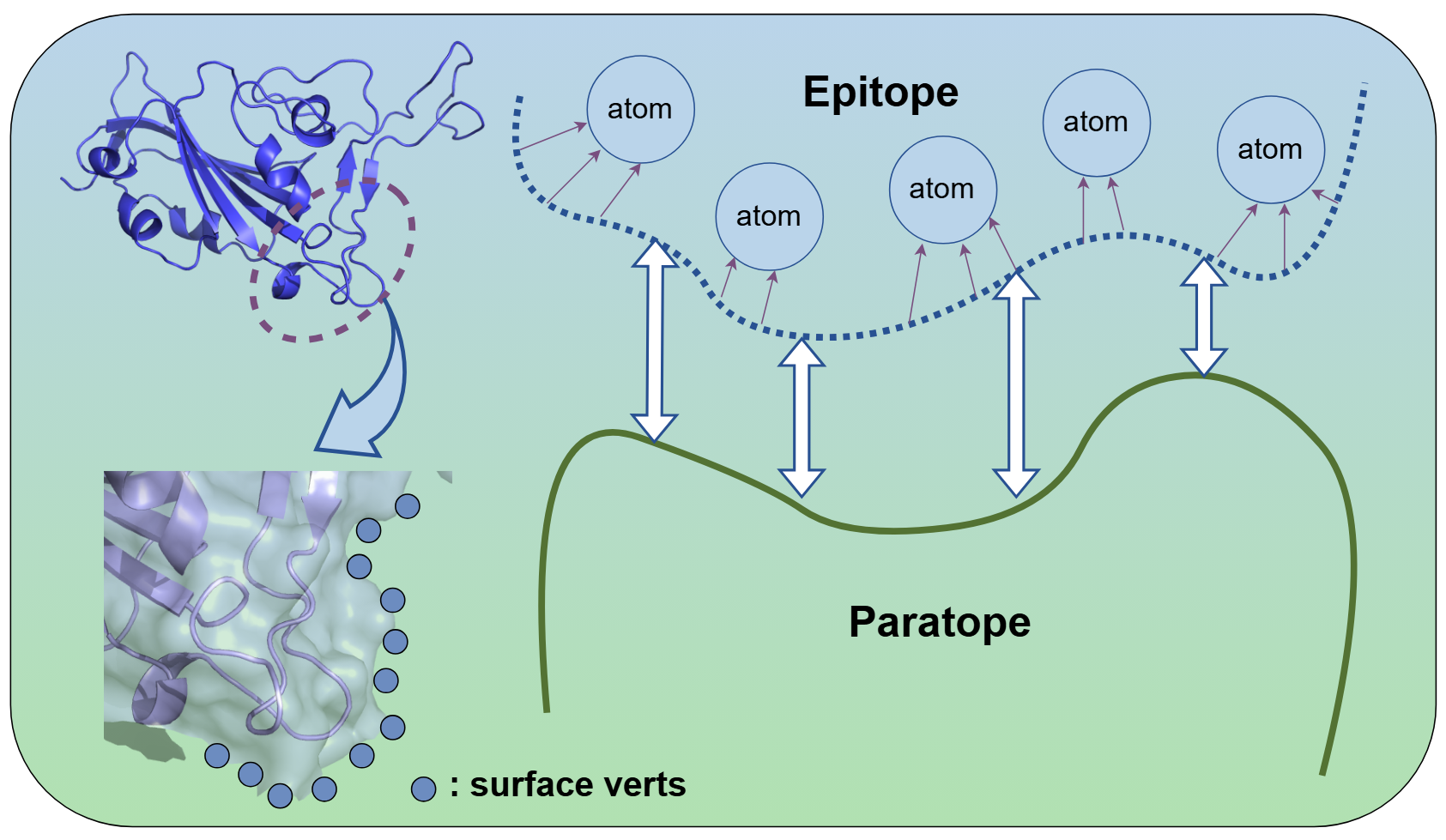}}
\caption{\textbf{Schematic diagram of surface definition and message passing.} We generate surface vertices within a cutoff of 1.5 Å around the epitope. For each vertex, we assign the nearest residue. Then, we perform message passing between the surface vertices and the antibody paratope using SME to enhance the learning ability of the interaction.}
\label{fig6}
\end{center}
\Description{Schematic diagram of surface definition and message passing.}
\end{figure}
\section{Sampling and Generation}
\label{sec:generation}

Algorithm~\ref{alg2} summarizes the one-step sampling procedure of AbFlow.

\begin{algorithm}[t]
\caption{Sampling Algorithm of AbFlow}
\label{alg2}
\begin{algorithmic}[1]
\State \textbf{Input:} condition $C$ and trained model $v_{\theta}$
\State Sample $Z_0=(S_0,X_0)\sim p_0(Z\mid C)$ at $t=0$
\State $\hat{X}_1\leftarrow v_{\theta}^{x}(Z_0,C)$
\State $\hat{q}_1\leftarrow\operatorname{Softmax}\left(v_{\theta}^{s}(Z_0,C)\right)$
\State $\hat{S}_1\leftarrow\operatorname*{arg\,max}\hat{q}_1$
\State \Return local sequence and structural endpoints $(\hat{S}_1,\hat{X}_1)$
\end{algorithmic}
\end{algorithm}

Within the same forward process, the predicted paratope information is propagated across the antibody through the global message-passing component of $v_{\theta}$ to generate the complete antibody sequence and full-atom structure. The generated antibody is then rigidly aligned with the predicted local interface before assembling the final antigen--antibody complex.

\section{Losses}
\label{l&g}

In this section, we provide detailed definitions of the sequence loss $\mathcal{L}_{\mathrm{seq}}$, structure loss $\mathcal{L}_{\mathrm{struct}}$, and antigen-paratope edge-distance loss $\mathcal{L}_{\mathrm{edge}}$ used in Equation~\eqref{LOSS}. The interface loss and the complete docking objective are defined in Section~\ref{sec:train}.

The sequence loss calculates the cross-entropy between the predicted amino acid distribution $p_i^{(r)}$ at each message-passing round and the one-hot vector $p_i^{\mathrm{true}}$ of the ground-truth residue type:
\begin{equation}
\mathcal{L}_{\mathrm{seq}}
=
\frac{1}{R|\mathcal{V}_{P}|}
\sum_{r=1}^{R}
\sum_{i\in\mathcal{V}_{P}}
\ell_{\mathrm{ce}}
\left(
p_i^{(r)},p_i^{\mathrm{true}}
\right).
\label{seq loss}
\end{equation}

The structure loss consists of the coordinate loss and the chemical-bond loss, which measure global and local structural deviations, respectively. The coordinate loss calculates the Huber loss between the final prediction $X_i^{(R)}$ and the ground-truth coordinates $X_i^{\mathrm{true}}$ after alignment. The chemical-bond loss calculates the Huber loss between the predicted and ground-truth bond lengths. Thus, $\mathcal{L}_{\mathrm{struct}}=\mathcal{L}_{\mathrm{coord}}+\mathcal{L}_{\mathrm{bond}}$:
\begin{equation}
\mathcal{L}_{\mathrm{coord}}
=
\frac{1}{|\mathcal{V}_{A}|}
\sum_{v_i\in\mathcal{V}_{A}}
\ell_{\mathrm{huber}}
\left(
X_i^{(R)},X_i^{\mathrm{true}}
\right),
\label{coord loss 1}
\end{equation}
\begin{equation}
\mathcal{L}_{\mathrm{bond}}
=
\frac{1}{|\mathcal{B}|}
\sum_{b\in\mathcal{B}}
\ell_{\mathrm{huber}}
\left(
b^{(R)},b^{\mathrm{true}}
\right).
\label{coord loss 2}
\end{equation}

The edge-distance loss supervises the predicted residue distances across the antigen-paratope interface at each message-passing round:
\begin{equation}
\mathcal{L}_{\mathrm{edge}}
=
\frac{1}{|\mathcal{E}_{EP}|}
\sum_{r=1}^{R}
\sum_{(u,v)\in\mathcal{E}_{EP}}
\ell_{\mathrm{huber}}
\left(
d^{(r)}(u,v),d^{\mathrm{true}}(u,v)
\right),
\label{dock loss 2}
\end{equation}
where $\mathcal{E}_{EP}$ denotes the interaction edges between the epitope and paratope, and $d^{(r)}(u,v)$ and $d^{\mathrm{true}}(u,v)$ denote the predicted and ground-truth distances between residues $u$ and $v$, respectively.


\section{Standard Multi-step Flow Matching}
\label{app:standard_fm}

To compare with the one-step formulation adopted in AbFlow, we further implement a standard multi-step flow-matching variant while keeping the same SME, message-passing backbone, conditional information, and auxiliary objectives.

Let $X_0$ and $X_1$ denote the initialized and ground-truth paratope coordinates, respectively. Before constructing the flow path, we align $X_0$ to $X_1$ using Kabsch alignment \citep{kabsch1976solution} followed by residue-level Hungarian matching \citep{kuhn1955hungarian} based on $C_{\alpha}$ distances, and denote the aligned coordinates as $\widetilde{X}_0$.

For each training sample, we draw $t\sim\mathcal{U}(0,1)$ and define
\begin{equation}
    X_t=\sigma_t\widetilde{X}_0+tX_1,
    \qquad
    \sigma_t=1-(1-\sigma_{\min})t,
\end{equation}
where $\sigma_{\min}=0.01$. The corresponding target coordinate velocity is
\begin{equation}
    u_x^*=X_1-(1-\sigma_{\min})\widetilde{X}_0.
\end{equation}

For the sequence state, we interpolate between the one-hot \texttt{[MASK]} state $q_0$ and the ground-truth amino-acid state $q_1$:
\begin{equation}
    S_t=(1-t)q_0+tq_1,
    \qquad
    u_s^*=q_1-q_0.
\end{equation}

The network explicitly takes $(X_t,S_t,t)$ as input and predicts the coordinate and sequence velocities $v_\theta^x$ and $v_\theta^s$. The flow-matching losses are
\begin{equation}
    \mathcal{L}_{\mathrm{FM}}
    =
    \ell_{\mathrm{Huber}}(v_\theta^x,u_x^*)
    +\ell_{\mathrm{Huber}}(v_\theta^s,u_s^*).
\end{equation}
The complete training objective is
\begin{equation}
    \mathcal{L}
    =
    \mathcal{L}_{\mathrm{seq}}
    +
    \mathcal{L}_{\mathrm{struct}}
    +
    \mathcal{L}_{\mathrm{dock}}
    +
    \mathcal{L}_{\mathrm{FM}}.
\end{equation}

At inference time, the coordinate and sequence states are initialized in the same manner as during training. We use explicit Euler integration with $N$ steps:
\begin{align}
    X_{k+1}
    &=
    X_k+\Delta t\,v_\theta^x(X_k,S_k,t_k),\\
    S_{k+1}
    &=
    S_k+\Delta t\,v_\theta^s(X_k,S_k,t_k),
\end{align}
where $\Delta t=1/N$ and $t_k=k/N$. After each sequence update, the state is clipped to non-negative values and normalized to the probability simplex.

Table~\ref{tab:standard_fm_steps} reports the effect of different numbers of Euler integration steps.

\begin{table}[h]
\centering
\caption{Effect of sampling steps for standard multi-step flow matching.}
\label{tab:standard_fm_steps}
\resizebox{\columnwidth}{!}{
\begin{tabular}{c|ccc|ccc}
\hline
\multirow{2}{*}{Steps}
& \multicolumn{3}{c|}{Overall}
& \multicolumn{3}{c}{Interface} \\
& AAR$\uparrow$
& TMscore$\uparrow$
& lDDT$\uparrow$
& CAAR$\uparrow$
& RMSD$\downarrow$
& DockQ$\uparrow$ \\
\hline
10 & 0.2789 & 0.9649 & 0.8256 & 0.1382 & 17.11 & 0.243 \\
20 & 0.2623 & 0.9647 & 0.8245 & 0.1315 & 17.09 & 0.243 \\
30 & 0.2586 & 0.9649 & 0.8243 & 0.1253 & 17.07 & 0.236 \\
50 & 0.2573 & 0.9649 & 0.8244 & 0.1202 & 17.07 & 0.241 \\
\hline
\end{tabular}}
\end{table}

Increasing the number of sampling steps does not improve performance, with the global structural metrics remaining nearly unchanged and the sequence-related metrics generally decreasing. We therefore use 10 sampling steps for the Standard FM result reported in Table~\ref{abl1}.

\section{Hyperparameter and Hardware Requirements}
\label{sec:hard}
Our AbFlow is trained using the Adam optimizer within the PyTorch data-parallel framework, implemented on 2 NVIDIA A100 GPUs, each with 80GB of VRAM. The hyperparameters of the model are detailed in Table \ref{hyp}.

\begin{table}[H]
\centering
\caption{Model hyperparameters and training configurations}
\label{hyp}
\begin{tabular}{lcl}
\toprule
\textbf{Hyperparameter} & \textbf{Value} & \textbf{Description} \\
\midrule
embed\_size       & 64  & Sequence embedding size. \\
hidden\_size      & 128 & Hidden state dimension. \\
n\_layers         & 3   & One-step velocity network layers. \\
k\_neighbors      & 9   & KNN neighbors per node. \\
batch\_size       & 16  & Training batch size. \\
$M_0$             & 50  & Surface vertices per residue. \\
\bottomrule
\end{tabular}
\end{table}
\end{document}